\pdfoutput=1 

\newif\ifconference
\conferencefalse

\documentclass[11pt]{article}
\usepackage[utf8]{inputenc}
\usepackage[english]{babel}
\usepackage{amsmath,amsfonts}
\usepackage{amssymb}
\usepackage[showonlyrefs]{mathtools}
\usepackage{amsthm}
\usepackage{moresize}
\usepackage[letterpaper, portrait, left=1in, right=1in, top=1in, bottom=1in]{geometry}
\usepackage[numbers,sectionbib,longnamesfirst]{natbib}
\usepackage[hyperindex,breaklinks,bookmarks]{hyperref}
\hypersetup{bookmarksdepth=2}
\usepackage[shortlabels]{enumitem}
\usepackage[toc]{appendix}
\usepackage{microtype}
\usepackage[ruled,linesnumbered]{algorithm2e}
\usepackage{cleveref}
\usepackage{xcolor}
\usepackage{apxproof}
\allowdisplaybreaks

\SetKwFor{RepTimes}{repeat}{times}{end}

\newtheorem{theorem}{Theorem}
\newtheorem{corollary}[theorem]{Corollary}
\newtheorem{definition}[theorem]{Definition}
\newtheorem{lemma}[theorem]{Lemma}

\newtheorem{fact}[theorem]{Fact}

\theoremstyle{remark}

\date{}
\title{Additive Noise Mechanisms for Making Randomized Approximation Algorithms Differentially Private
}

\author{
	Jakub Tětek\\
	\texttt{j.tetek@gmail.com}\\
	Basic Algorithms Research Copenhagen,\\
        University of Copenhagen
}

\begin{document}
\maketitle

\begin{abstract}
The exponential increase in the amount of available data makes taking advantage of them without violating users' privacy one of the fundamental problems of computer science.
This question has been investigated thoroughly under the framework of differential privacy. However, most of the literature has not focused on settings where the amount of data is so large that we are not even able to compute the exact answer in the non-private setting (such as in the streaming setting, sublinear-time setting, etc.). This can often make the use of differential privacy unfeasible in practice.

In this paper, we show a general approach for making Monte-Carlo randomized approximation algorithms differentially private. We only need to assume the error $R$ of the approximation algorithm is sufficiently concentrated around $0$ (e.g.\ $\mathbb{E}[|R|]$ is bounded) and that the function being approximated has a small global sensitivity $\Delta$. Specifically, if we have a randomized approximation algorithm with sufficiently concentrated error which has time/space/query complexity $T(n,\rho)$ with $\rho$ being an accuracy parameter, we can generally speaking get an algorithm with the same accuracy and complexity $T(n,\Theta(\epsilon \rho))$ that is $\epsilon$-differentially private.

Our technical results are as follows. First, we show that if the error is subexponential, then the Laplace mechanism with error magnitude proportional to the sum of the global sensitivity $\Delta$ and the \emph{subexponential diameter} of the error of the algorithm makes the algorithm differentially private. This is true even if the worst-case global sensitivity of the algorithm is large or infinite. We then introduce a new additive noise mechanism, which we call the zero-symmetric Pareto mechanism. We show that using this mechanism, we can make an algorithm differentially private even if we only assume a bound on the first absolute moment of the error $\mathbb{E}[|R|]$.

Finally, we use our results to give either the first known or improved sublinear-complexity differentially private algorithms for various problems. This includes results for frequency moments, estimating the average degree of a graph in sublinear time, rank queries, or estimating the size of the maximum matching. 
Our results raise many new questions and we state multiple open problems.
\end{abstract}
 \thispagestyle{empty}
\newpage
\clearpage
\setcounter{page}{1}

\section{Introduction}
With the increase in the amount of available data, the problem of analyzing it in a privacy-preserving manner has become a central problem in computer science. 
One commonly used tool for this task is
differential privacy, which is a well-established notion of privacy that is commonly used in data analysis and machine learning. 
However, with some notable exceptions, the literature on differential privacy has focused on the setting where the amount of data is small enough that we would be able to practically solve a given problem exactly in the non-private setting.
However, in practice, this assumption is often not realistic -- this is after all the reason for the existence of (among others) streaming and sublinear-time algorithms. 

Our goal is thus to get very efficient (sublinear) algorithms that at the same time guarantee differential privacy. In the streaming or sublinear-time setting, the error coming from the algorithm not being exact will generally speaking be much bigger than the amount of noise that the given problem necessitates for ensuring differential privacy.
The main objective in this setting is thus not to simply minimize the amount of noise we add, but rather to achieve a given level of accuracy while
minimizing the complexity (e.g.~space, time, or query complexity) of the algorithm. Of course, some amount of noise inherently has to be added to achieve privacy, but this is usually so small, that one would need linear complexity to get such a level of accuracy even without privacy. In the sublinear regime, we thus usually do not have to worry whether a given level of accuracy is achievable and we instead focus as our central objective on the complexity needed to achieve it.

One of the main difficulties in making sublinear algorithms private is that most sublinear-time and streaming algorithms are randomized and give only probabilistic guarantees on the quality of the output.
This makes adding noise based on global sensitivity\footnote{Global sensitivity of a function $g$ with respect to a relation $\sim$ is defined as $\sup_{x \sim x'} |g(x) - g(x')|$.} -- which is commonly used to get differentially private algorithms -- unsuitable for this situation, as in the worst case the global sensitivity of the approximation algorithm\footnote{Here, we see the approximation algorithm as a deterministic function of the input and a string of random bits.} can be very large even if the global sensitivity of the function being approximated is small. In this paper, we propose a way to get around this issue by showing additive noise mechanisms that only need that $(1)$ the function being approximated has low global sensitivity and $(2)$ the answer of the algorithm is sufficiently concentrated around the true value.

We give two main results -- one under the assumption that the error has subexponential tails, while the other only assumes bounded mean deviation (or higher moments). While the first one has a much stronger assumption about the error distribution, it is stronger in that it also works for multiple adaptive queries that are \emph{not answered independently}. This is useful for example for streaming algorithms, where multiple queries can be answered using a single sketch and are thus not answered independently. Note that the standard composition theorem would allow us to perform multiple queries only if they were answered independently.

We use our results to give new differentially private algorithms for various problems: for maximum matching under node-level privacy, frequency moments, counting connected components under edge-level privacy, and rank queries. 
We also show how a common technique for designing relative-approximation sublinear-time algorithms -- \emph{advice removal by geometric search} -- can be made differentially private. This implies an edge-differentially-private algorithm for estimating the average degree of a graph, improving upon the state of the art \cite{blocki2022privately}, but we think it could also be useful for many other problems. Our algorithm for maximum matching also answers an open problem from \cite{blocki2022privately}. For the other mentioned problems, we either give the first sublinear-complexity algorithm that ensures \emph{pure} differential privacy, or one that is more efficient.

\subsection{High-level view of technical results}
We now informally state our main technical results and sketch how we use them in order to get private algorithms.
\vspace{-1em}
\paragraph{Subexponential error tails.} Our first central result is as follows: 
\smallskip
\\ \noindent \textbf{\Cref{thm:subexponential}, simplified version.} \begin{em}
Assume we have an algorithm $A(D)$ for $D$ being a dataset. Assume there exists a function $g$ with global sensitivity $\leq \Delta_1$ w.r.t.\ $D$ such that $A(D) - g(D)$ has \emph{subexponential diameter} $\leq \Delta_2$, i.e.\ $\mathbb{P}[|A(D) - g(D)| \geq t] \leq 2e^{-t/\Delta_2}$ for $t \geq 0$. Then releasing $A(D) + \text{Laplace}\big(O\big((\Delta_1 + \Delta_2) /\epsilon\big)\big)$ is $\epsilon$-differentially private for $\epsilon \leq O(1)$.

Moreover, with noise $\text{Laplace}\big(O\big(k (\Delta_1 + \Delta_2) /\epsilon\big)\big)$, this also holds if we make $k$ such releases with different algorithms $A_1, \cdots, A_k$ chosen adaptively that are executed with the same randomness.
\end{em}\medskip\\ \noindent
Note that this generalizes the claim that the Laplace mechanism with noise magnitude proportional to the global sensitivity gives differential privacy (this can be seen by setting $A(D) = g(D)$). Note also that in the second half, the algorithms use the same randomness, and we thus cannot get this part of the result by standard composition.

One would usually use this theorem as follows.
We start with a randomized approximation algorithm whose error is subexponentially concentrated around zero (often, this is either known or easy to prove) that approximates a parameter with a small global sensitivity $\Delta_1$. This is the case for example for the YYI maximum matching algorithm \cite{yoshida2009improved} under node-level privacy or the KLL sketch \cite{karnin2016optimal} for rank queries. Suppose the complexity (such as time/space/query/sample or other complexity) of the algorithm is $T(n,\rho)$ and has additive error of scale (i.e.\ with subexponential diameter) $\rho n$. If we want the final error with privacy to be $O(\rho n)$, then we run the algorithm with error parameter $\epsilon \rho$, making the error's subexponential diameter be $O(\epsilon \rho n)$. We then get from \Cref{thm:subexponential} that adding noise of magnitude $O(\rho n)$ is sufficient to get $\epsilon$-differential privacy, assuming $\Delta_1$ is sufficiently small, giving us the desired result. The complexity of the private algorithm will thus be $T(n, \Theta(\epsilon \rho))$. This allows us to achieve a given level of accuracy\footnote{This is true unless $\rho$ is very small, as otherwise the global sensitivity will necessitate some level of noise. As we noted, this usually happens only when $T(n, \epsilon \rho) \geq \Omega(n)$, making this uninteresting for our sublinear setting.} under pure differential privacy, while not significantly worsening the algorithm's complexity. We describe this approach in greater detail and with general failure probabilities (not just constant) in \Cref{sec:applications}.

To illustrate the second part of the theorem, consider for example a rank queries sketch (see \Cref{sec:rank} for details). The algorithms $A_1, \cdots, A_k$ correspond to making $k$ adaptive queries to the sketch of a dataset (assume we are given $k$ queries we need to perform) and the fact that the algorithms use the same randomness corresponds to us querying the same sketch (as compared to $k$ independent sketches). Specifically, the algorithm $A_i$ here builds the sketch (using the shared randomness), and then performs the $i$-th query. Note that the fact that we only use one sketch \emph{prevents us from using the composition theorem} to get this from just the first part of the theorem.

\vspace{-1em}
\paragraph{Polynomial error tails.} We then prove that in the case of a single query, it is sufficient to assume a bound on some deviation moments (multiple independently answered queries can be handled using the standard composition theorem). 
\smallskip
\\ \noindent \textbf{\Cref{thm:polynomial}, simplified version.} \begin{em}
Let us have an algorithm $A$ such that there exists a function $g$ with global sensitivity $\leq \Delta$ and such that $\mathbb{E}[|A(D) - g(D)|^3] \leq \Delta^3$ for any dataset $D$. Then for $\epsilon \leq O(1)$ there exists a random variable $Y$ with $\mathbb{E}[|Y|] \leq O(\Delta/\epsilon)$, such that $A(\cdot) + Y$ is $\epsilon$-differentially private.
\end{em} \medskip\\ \noindent
This theorem can be used in a way similar to \Cref{thm:subexponential}, whose use we described above. Moreover, it holds that if we have $\mathbb{E}[|A(D) - g(D)|] \leq \Delta$, then by taking a median of $5$ independent executions of $A$, we get an algorithm $A'$ whose error's third moment is also bounded: $\mathbb{E}[|A'(D) - g(D)|^3] \leq O(\Delta^3)$ \cite{larsen2021countsketches}. This allows us to use the theorem even if we only have a bound on $\mathbb{E}[|A(D) - g(D)|]$ or the mean squared error $\mathbb{E}[(A(D) - g(D))^2]$. The approach to using \Cref{thm:subexponential} we described above can with \Cref{thm:polynomial} give the following ``nicely packaged version'' of the theorem:
\smallskip
\\ \noindent \textbf{\Cref{lem:privatize_poly}, simplified version.} \begin{em}
Suppose there is an algorithm $A$ approximating a function $g$ with global sensitivity $\Delta$ such that $E[|A(x) - g(x)|] \leq \rho f(x)$ for some function $f$ with  time/space/query complexity $T(n, \rho)$. Then for $\epsilon \leq O(1)$ there exists an $\epsilon$-differentially private algorithm $A'$ such that when $\epsilon \rho \geq \Omega(\Delta /f(x))$, it holds $P[|A'(x) - g(x)| > \rho f(x)] \leq 1/3$ and complexity $O(T(n, \epsilon \rho))$.
\end{em} \medskip\\ \noindent
The failure probability $1/3$ can be decreased by standard probability amplification.

The problem of differentially private randomized approximation algorithms has been explored independently of this work by \citet{blocki2022make}.
The techniques used in \cite{blocki2022make} differ significantly from those used in this paper. Specifically, in \cite{blocki2022make}, the authors set the failure probability of the algorithm to be $\leq \delta$ (for example by probability amplification), thus limiting the global sensitivity of the algorithm up to an event of probability $\leq \delta$. This allows them to rely on the standard result for getting differential privacy based on global sensitivity\footnote{They also consider functions with low smoothed sensitivity instead of just low global sensitivity; we do not consider that in this paper.}. Specifically, if the algorithm has complexity $T(n, \rho)$ and one uses probability amplification, then the approach of \cite{blocki2022make} gives an $(\epsilon,\delta)$-differential privacy in complexity $O(T(n, \epsilon \rho) \log \delta^{-1})$.
In this paper, instead of relying just on global sensitivity, we instead prove privacy from first principles. At the cost of assuming that the error is sufficiently concentrated, we show that the probability amplification step is not needed, allowing us to get $\epsilon$-differential privacy in the better complexity of $T(n, \epsilon \rho)$.

Our approach allows us to give more efficient algorithms than \citet{blocki2022make} for several problems: estimating the average degree of a graph, estimating the size of a maximum matching, and estimating the number of connected components; 
\emph{for all these problems we get pure differential privacy {instead of approximate $(\epsilon,\delta)$-privacy} while saving a $\log \delta^{-1}$ factor in the complexity}.

\subsection{Our techniques}
\paragraph{Subexponential error, one query.}
Suppose we have a randomized algorithm $A(D)$ for $D$ being a dataset that approximates a function $g(D)$ with global sensitivity $\leq \Delta_1$ for some parameter $\Delta_1$. Define the error $R$ as the random variable $R = A(D) - g(D)$ and assume that it is tightly concentrated around $0$, namely $\mathbb{P}[|R| > t] \leq 2e^{-t/\Delta_2}$ for some value $\Delta_2$ ($\Delta_2$ is an upper bound on the ``subexponential diameter" of $R$). Intuitively speaking, $\Delta_2$ determines the ``scale" of $R$, and $\Delta_2$ is in fact up to a constant factor an upper bound on $\mathbb{E}[|R|]$. Note that the tails of $R$ decrease at least at the same rate as those of the Laplace distribution. This suggests Laplacian noise with large enough magnitude could ``hide $R$''. 
Indeed, we prove that Laplacian noise will guarantee privacy. We now sketch the proof.


\medskip \noindent
\textit{High-level view.} Assume for simplicity that both $\Delta_1, \Delta_2 \leq 1$. The same approach works for general values of $\Delta_1, \Delta_2$ by simple re-scaling. Let $Y \sim \text{Laplace}(c/\epsilon)$ for appropriately chosen value of $c$. We will prove that for any random variable $X$ with subexponential diameter $\leq 3$ \footnote{We choose value $3$ as this is the value we will need below. The claim holds also for larger constants with $c$ chosen appropriately.} (that is $\mathbb{P}[|X| \geq t] \leq 2e^{-t/3}$), the probability density functions satisfy $f_{X+Y}(y)/f_Y(y) = e^{\pm \Theta(\epsilon)}$ for any $y$. We can use this to prove privacy, as we now show. For two neighboring datasets $D_1,D_2$, we set $R_1 = A(D_1) - g(D_1)$ and $R_2 = A(D_2) - g(D_1)$ (note the asymmetry in the definitions). It then holds that $R_1$ has subexponential diameter $\Delta_2 \leq 1$. One can also show that the subexponential diameter of $R_2 = (A(D_2)-g(D_2)) + (g(D_2) - g(D_1))$ is $\leq 3$ (\Cref{lem:sub_inequality}). Let $y' = y-g(D_1)$. It then holds for any $y$ that
\[
\frac{f_{A(D_1)+Y}(y)}{f_{A(D_2) + Y}(y)} = \frac{f_{g(D_1) + R_1+Y}(y)}{f_{g(D_1) + R_2 + Y}(y)} =
\frac{f_{R_1+Y}(y')}{f_{R_2 + Y}(y')} = \frac{f_{R_1+Y}(y')}{f_{Y}(y')}\cdot \frac{f_{Y}(y')}{f_{R_2 + Y}(y')} = e^{\pm \Theta(\epsilon)}
\]
where the last equality uses the claim $f_{X+Y}(y')/f_Y(y') = e^{\pm \Theta(\epsilon)}$ for $X = R_1$ and for $X = R_2$. This implies differential privacy.

\medskip \noindent
\textit{Bounding ratios of density functions.} We now sketch why $f_{X+Y}(y)/f_Y(y) = e^{\pm \Theta(\epsilon)}$. Since $Y$ is continuous, we may re-write
\[
f_{X + Y}(y) = \mathbb{E}[f_Y(y-X)] = \frac{\epsilon}{2c} \mathbb{E}[\exp(-\epsilon |X - y|/c)] = (*)
\]
where the first equality is a standard identity \cite{density_indentity}.
We now use the inequalities $|X-y| \leq |X| + |y|$ and $|X-y| \geq |y| - |X|$. This allows to bound
\begin{align*}
(*) \leq& \frac{\epsilon}{2c} \mathbb{E}[\exp(-\epsilon (|y| - |X|)/c)] =  \frac{\epsilon}{2c} e^{-\epsilon|y|/c} \mathbb{E}[\exp(\epsilon |X|/c)]\\
(*) \geq& \frac{\epsilon}{2c} \mathbb{E}[\exp(-\epsilon (|X| + |y|)/c)] =  \frac{\epsilon}{2c} e^{-\epsilon|y|/c} \mathbb{E}[\exp(-\epsilon |X|/c)]
\end{align*}
while it holds that $f_Y(y) = \frac{\epsilon}{2c} e^{-\epsilon|y|/c}$. It is thus sufficient to prove that $\mathbb{E}[\exp(\epsilon |X|/c)] \leq e^{O(\epsilon)}$ and $\mathbb{E}[\exp(-\epsilon |X|/c)] \geq e^{-O(\epsilon)}$. While the first inequality is standard, the second is not. We will now sketch a proof for both.

\medskip \noindent
\textit{Bounding the expectation of exponentials of a subexponential random variable.} 
If we knew the density function of $X$, we could easily express the expectations as integrals. However, not only we do not have a bound on the density, but the density may even not exist. We thus use the following trick. We use the fact that for any real random variable $Z$, it holds that $Z$ has the same distribution as $F_Z^{-1}(u)$ for $u \sim Unif(0,1)$ where $F_Z$ is the cumulative distribution function (CDF) of $Z$. This allows us to write $\mathbb{E}[e^{-\epsilon|X|/c}] = E_u[e^{-F^{-1}_{\epsilon|X|/c} (u)}]$ and similarly for $\mathbb{E}[e^{\epsilon|X|/c}]$.

Unlike the density function, we do have a bound on the cumulative distribution function. Specifically, we are assuming $\mathbb{P}[|X| > t] \leq 2e^{-t/3}$ which implies that $F_{\epsilon |X|/c}^{-1}(u) \leq -\frac{3 \epsilon}{c} \log(\frac{1-u}{2})$. Upper-bounding the CDF like this reduces the problem to computing the expectation of a function of the uniform random variable, which can be done straightforwardly. This proves the desired bounds.

%
%

\paragraph{Subexponential error, multiple queries.} We would like to be able to release answers to multiple queries which are not answered independently (such as if they are answered based on the same sketch). Since the answers are not independent, we cannot use the composition theorem. We now sketch an alternative approach.

For fixed queries, the above proof goes through with minor modifications even in the multivariate case. Instead of using the inequalities $|y-X| \leq |y| + |X|$ and $|y-X| \geq |y| - |X|$ in the case of $y,X \in \mathbb{R}$, we use the analogous version for $\ell_1$ norms in the case of $y,X \in \mathbb{R}^k$: $\|y-X\|_1 \leq \|y\|_1 + \|X\|_1$ and $\|y-X\|_1 \geq \|y\|_1 - \|X\|_1$. We then use that if $X$ is a vector of $k$ subexponential random variables with diameter $\leq \Delta$, then $\|X\|_1$ has subexponential diameter $\leq 3 k \Delta$ (\Cref{lem:sub_inequality}).

This however gives the result only in the nonadaptive case, when the queries do not depend on the released values: the identity $f_{X+Y}(y) = \mathbb{E}[f_Y(X-y)]$ relies crucially on $X = (X_1, \cdots, X_k)$ and $Y = (Y_1, \cdots, Y_k)$ being independent.
In the case of adaptive queries, $X_i$ could depend on $Y_1, \cdots, Y_{i-1}$ (the query that we perform -- and thus also the answer to it -- can be influenced by the noise we add to the previous answers). We instead use our non-adaptive version of the claim and make it adaptive in a black box fashion, by proving a claim that may be of independent interest: If we have a countable number of mechanisms such that releasing the answers of any \emph{fixed} subset of size $k$ is $\epsilon$-differentially private, then we may also pick this subset \emph{adaptively} and it will still be $\epsilon$-differentially private.

\paragraph{Error with polynomial tails.}
In the case that the error has polynomial tails, we only consider the case of a single query. Our techniques do not seem to generalize to the multivariate case, and we conjecture that this is impossible (see \Cref{sec:open_problems}). The case when multiple queries are answered independently may be still handled by the standard composition theorem.

An approach similar to the one described above can be made to work, with the difference that we use the inequality $|x - X| \geq \max(0, |x|-|X|)$ instead of the weaker $|x-X| \geq |x|-|X|$. This approach, however, requires proving the following inequality for all $y,s \geq 0, \alpha > 1, 0 \leq \epsilon \leq 1$:
\begin{align*}
\int_0^1 \min\left((1+|y|/s)^\alpha,\left|1- \frac{(1-2^{-1/\alpha}) \epsilon}{(1+|y|/s)(1-u)^{1/\alpha}}\right|^{-\alpha}\right) du \leq 1+\frac{2\alpha-1}{\alpha-1} \epsilon.
\end{align*}
This is the technically most challenging part of this paper. The trick is to bound the inside of the integral for $u \in [0,1-\epsilon (1+|y|/s)^{-\alpha}]$ by a simpler expression that can be successfully integrated. The rest of the interval $[0,1]$ contributes at most $\epsilon$, as its length is $\epsilon (1+|y|/s)^{-\alpha}$ and the maximum value of the function being integrated is $\leq (1+|y|/s)^\alpha$.

\subsection{Related work}

To the best of our knowledge, the work on differentially private approximation algorithms started with private sketches. \citet{mir2011pan} gave pan-private\footnote{An algorithm on an input stream is said to be pan-private if releasing the internal state of the algorithm at any point in the computation is differentially private. It is a strictly stronger notation than differential privacy of the output.} sketches for heavy hitters. An improved sketch has been recently given by \citet{pagh2022improved}. A private version of the deterministic Misra-Gries sketch \cite{misra1982finding} for heavy hitters has been recently given by \citet{lebeda2023better}. Heavy hitters were also investigated in the multi-party computation setting \cite{huang2021frequency}, in the local differential privacy setting \cite{bassily2017practical}, and using cryptographic assumptions \cite{melis2015efficient,ghazi2021power,huang2021frequency}. 

A sketch for fractional frequency moments $F_p$ for $0 \leq p \leq 1$ has been given by \citet{wang2021differentially}. After releasing this paper, \citet{epasto2023differentially} have given an algorithm general value of $p$ in the continual release setting. A sketch for differentially private quantiles has been given by \citet{alabi2022bounded}. A technique for stream sanitization has been given by \citet{kaplan2021note}; this work resulted in improved differentially private sketches for approximate quantiles. An approach for differentially privately estimating distances in euclidean spaces using private sketches has been given by \citet{stausholm2021improved}. A general approach to making linear sketches differentially private was given by \citet{zhao2022differentially}.

A recent line of work has shown that many sketches already provide privacy \emph{by themeselves} or with only small modifications, without adding any noise. \citet{blocki2012johnson} have shown that the Johnson-Lindenstrauss transform by itself ensures differential privacy.
\citet{smith2020flajolet} have shown that this is also the case for the Flajolet-Martin sketch for counting distinct elements and \citet{choi2020differentially} have given a similar result for the LogLog algorithm \cite{durand2003loglog}. This was recently generalized by \citet{dickens2022nearly} who show that this is not only the case for the two above-mentioned sketches, but in fact for a large class of sketches for counting distinct elements. 

As far as sublinear-time algorithms are concerned, \citet{sivasubramaniamdifferentially} have shown a differentially private algorithm that returns a $2+ \rho$ approximation of the number of edges in a graph in time $\tilde{O}_{\rho,\epsilon}(\sqrt{n})$. This has been later improved by \citet{blocki2022privately} to $1+\rho$ approximation in the same complexity. In that paper, the authors also give differentially private sublinear algorithms for approximate maximum matching and vertex cover. A sublinear time algorithm for estimating the median was recently given by \citet{boehler2022secure}.


\section{Preliminaries} \label{sec:preliminaries}
\subsection{Differential privacy}
Throughout the paper, we assume that we have a symmetric ``neighbor" relation $\sim$ on the set of all possible datasets. Intuitively speaking, in the case when we have a database of users, this should correspond to two datasets being the same except for the data of one user whose privacy we are trying to protect.
\begin{definition}[\cite{dwork2006differential}]
A randomized algorithm $M$ with range $S$ is $\epsilon$-differentially private if for any measurable $T \subseteq S$, it holds for any $x \sim x'$ for a symmetric ``neighbor'' relation $\sim$, that
\begin{align*}
e^{-\epsilon} \leq \frac{\mathbb{P}[M(x) \in T]}{\mathbb{P}[M(x') \in T]} \leq e^\epsilon
\end{align*}
\end{definition}

This definition is commonly relaxed to a notion called approximate differential privacy, with the above notion then being called pure privacy. In this paper, we will focus only on pure differential privacy.

If the output of $M$ is a continuous random variable, then it is sufficient to prove that for any $y$ and $x \sim x'$ it holds $e^{-\epsilon}\leq f_{M(x)}(y)/f_{M(x')}(y) \leq e^\epsilon$, where $f_X$ for $X$ being a continuous random variable is the probability density function of $X$.

The global sensitivity \cite{dwork2006calibrating} of a function $g$ is defined as
\[
\sup_{x \sim x'} |g(x) - g(x')|.
\]
\citet{dwork2006calibrating} have shown that if $g$ has global sensitivity $\Delta$, then adding $\text{Laplace}(\Delta/\epsilon)$ provides $\epsilon$-differential privacy.

In the context of graph problems, one often speaks of a mechanism being edge-differentially private, or node-differentially private. These terms refer to the relation $\sim$ that is used. In the case of node-differential privacy, we have $G \sim G'$ iff one can get $G$ from $G'$ by deleting one vertex and the incident edges, or the other way around. In the case of edge-differential privacy, we have $G \sim G'$ iff one can get $G$ from $G'$ by deleting one edge, or the other way around. 

\subsection{Probability theory} \label{sec:prelim_concentration}

If $D$ is a distribution, we use $D^{\otimes k}$ for $k$ being a natural number, to denote the $k$-fold product distribution of $D$. For a random variable $Z$, we denote by $F_Z$ its cumulative distribution function. We denote by $F^{-1}_Z(p) =\inf \{x \in \mathbb{R}: F_Z(x) \geq p\}$ its generalized inverse. It holds that $F^{-1}_Z(u)$ has the same distribution as $Z$ for $u \sim Unif(0,1)$ \cite{mcneil2015quantitative}.
We will need the following claim.
\begin{fact}[\cite{density_indentity}] \label{fact:important_identity}
Let $X,Y$ be independent random variables in $\mathbb{R}^k$, and assume $Y$ has a probability density function (pdf) $f_Y(z)$. Then the pdf of $X+Y$ is $f_{X+Y}(z) = \mathbb{E}[f_Y(z-X)]$.
\end{fact}
\noindent
We will also need the following claim:
\begin{lemma}[Moment amplification, \cite{larsen2021countsketches}] \label{lem:moment_amplification}
Let $X_1, \cdots X_{2k-1}$ be i.i.d.~random variables on $\mathbb{R}$. It holds for any $x,c \in \mathbb{R}$ that
\begin{align*}
E[|\text{median}(X_1, \cdots, X_{2k-1}) - x|^{c k}]^{1/ck} \leq O(E[|X_1-x|^c]^{1/c})
\end{align*}
\end{lemma}

\paragraph{Concentration of measure.} The notions of \emph{subexponential}
random variables and \emph{subexponential
diameter} $\sigma_{se}[X]$ are central to concentration of measure. 
There are several different definitions for $\sigma_{se}[X]$, that differ by constant factors. See for example \cite[Chapter 2]{vershynin2018high} for an exposition of the various definitions. In this paper, one of the definitions is especially suitable for the way we use it in our proofs, and that is the definition that we use. 

\begin{definition}
Let $X$ be a real random variable. We define the 
subexponential diameter of $X$, denoted by $\sigma_{se}[X]$ as the smallest values for which for any $t > 0$ holds
\begin{align*}
\mathbb{P}[|X| \geq t] \leq& 2\exp(- t/\sigma_{se}[X])
\end{align*}
A random variable $X$ is said to be 
subexponential if $\sigma_{se}[X] < \infty$. 
\end{definition}

It holds that $\sigma_{se}[c X] = c \sigma_{se}[X]$. 
It also holds that
\begin{fact} \label{fact:adding_constant}
For $X$ being a random variable and $c\geq 0$, it holds that $\sigma_{se}[X+c] \leq \sigma_{se}[X]+c/\log 2$.
\end{fact}
\begin{proof}
We can bound $\mathbb{P}[X+c \geq t] \leq \min(1,2e^{-(t-c)/\sigma_{se}[X]}) \leq \min(1, 2e^{-t/(\sigma_{se}[X]+c/\log 2)})$, thus implying the claim.
\end{proof}

Finally, the following is a standard claim, but we need the constant factor, which is specific to the definition of $\sigma_{se}$ that we are using. We thus give a proof in \Cref{sec:appendix}, based on the standard proof of triangle inequality for Orlicz norms.
\begin{lemma} \label{lem:sub_inequality}
Let us have real random variables $X_1, \cdots, X_k$. It holds
\[
\sigma_{se}[\sum_{i=1}^k X_i] \leq 3 \sum_{i=1}^k \sigma_{se}[X_i]
\]
\end{lemma}

\section{Algorithms with subexponential error} \label{sec:subexponential}
In this section, we show how algorithms, whose error has a small subexponential diameter, can be made differentially private.
We start by proving a technical lemma. We will later bound the ratios between probability densities of our mechanism's answers by the exponential expectations that we now bound and, finally, we will use that to prove privacy in \Cref{thm:subexponential}, which is the main theorem of this section.

Note that while the second of the two inequalities is standard, the first one is not. Our proof does not follow the strategy of the standard proof of the second inequality, which is based on a Taylor expansion and does not seem to straightforwardly apply to the first inequality.
\begin{lemma} \label{lem:exponential_expectations_bound}
Suppose $X$ is a random variable with subexponential diameter $\Delta \leq 1/2$. It holds $\mathbb{E}[e^{-|X|}] \geq \frac{2^{-\Delta}}{1+\Delta} \geq e^{-(1+\log 2)\Delta}$ and $\mathbb{E}[e^{|X|}] \leq \frac{2^\Delta}{1-\Delta} \leq e^{3 \log (2) \Delta}$.
\end{lemma}
\begin{proof}
Since $X$ is subexponential with diameter $\Delta$, it holds that $\mathbb{P}[|X| \geq z] \leq 2e^{-z/\Delta}$. Therefore, $F^{-1}_{|X|}(u) \leq -\Delta\log(\frac{1-u}{2})$. We use the fact that for $u \sim Unif(0,1)$, the random variable $F_{|X|}^{-1}(u)$ has the same distribution as $|X|$. We can now bound
\begin{align*}
\mathbb{E}[e^{-|X|}] =& E_u[e^{-F^{-1}_{|X|}(u)}] \\\geq& E_u[e^{\Delta\log(\frac{1-u}{2})}] \stepcounter{equation}\tag{\theequation}\label{eq:this_i_simplify} \\=&  2^{-\Delta} E_u[(1-u)^{\Delta}] \\=& 2^{-\Delta}\int_0^1 (1-u)^{\Delta} du \\=& 2^{-\Delta}\left[-\frac{(1-u)^{\Delta+1}}{\Delta+1}\right]_{u=0}^1 \\=& \frac{2^{-\Delta}}{1+\Delta} \geq e^{-(1+\log 2)\Delta}
\end{align*}
where the last inequality holds because we can equivalently write $2^{-\Delta}/(\Delta + 1) \geq (2 e)^{-\Delta}$, which simplifies to $e^\Delta \geq 1+\Delta$, which is a standard inequality. Similarly, we can bound $\mathbb{E}[e^{|X|}]$ as follows.
\begin{align*}
\mathbb{E}[e^{|X|}] =& E_u[e^{F^{-1}_{|X|}(u)}] \\\leq& E_u[e^{-\Delta\log(\frac{1-u}{2})}] \\=& 
\frac{2^\Delta}{1-\Delta}\leq e^{3 \log (2) \Delta}
\end{align*}
where the second equality is by substituting $\Delta$ with $-\Delta$ in \eqref{eq:this_i_simplify} (since as we have shown, \eqref{eq:this_i_simplify} is equal to $2^{-\Delta}/(1+\Delta)$).
We have here used that $\Delta < 1$ (otherwise the final expression may not be defined).
The last inequality can be shown as follows: we take the ratio of the two sides resulting in $h(\Delta) = 4^\Delta(1-\Delta)$ and we show $h(\Delta) \geq 1$ for $0 \leq \Delta \leq 1/2$. The function $h$ is concave (the second derivative is $-2 \cdot 2^{2\Delta} \log 4 + 2^{2\Delta} \log^2 4$, which can be easily seen to be negative), meaning that it is sufficient to check that the inequality holds at the endpoints of the interval $[0,1/2]$: that $h(0) \geq 1$ and $h(1/2) \geq 1$. One can easily check this holds. 
\end{proof}

We are now ready to prove a lemma about the ratio of the density function of a Laplace and of Laplace shifted by a subexponential random variable. We will then use this to prove differential privacy. Note that the random variables in the lemma do not have to be independent.
\begin{lemma} \label{lem:subexponential}
Let $X_1, \cdots, X_k$ be random variables with subexponential diameter at most $\Delta$ and let $X = (X_1, \cdots, X_k)$. Let $Y \sim Lap^{\otimes k}(k \Delta/\epsilon)$ for $\epsilon \leq 1/6$. Consider $y \in \mathbb{R}^k$. It holds $e^{- 3(1+\log 2) \epsilon} \leq f_{X+Y}(y)/f_Y(y) \leq e^{9\log(2) \epsilon}$. Moreover, if $k=1$ and $\epsilon \leq 1/2$, it holds $e^{- (1+\log 2) \epsilon} \leq f_{X+Y}(y)/f_Y(x) \leq e^{3\log(2) \epsilon}$.
\end{lemma}
\begin{proof}
By \Cref{fact:important_identity}, we have
\[
f_{X+Y}(y) = \mathbb{E}[f_Y(y-X)] = \left(\frac{\epsilon}{2k\Delta}\right)^k\mathbb{E}[\exp(-\frac{\epsilon}{k\Delta} \|X-y\|_1)]
\]
For the sake of brevity, we let $\gamma = \left(\frac{\epsilon}{2k\Delta}\right)^k$. We may bound $\|X-y\|_1 \leq \|X\|_1 + \|y\|_1$ and $\|X-y\|_1 \geq \|y\|_1 - \|X\|_1$. This allows us to bound 
\begin{align*}
f_{X+Y}(y) =& \gamma E\left[\exp\left(\frac{-\epsilon \|X-y\|_1}{k\Delta}\right)\right] \\\geq& \gamma E\left[\exp\left(\frac{-\epsilon (\|X\|_1 + \|y\|_1)}{k\Delta}\right)\right] \\=& \gamma \exp\left(\frac{-\epsilon \|y\|_1}{k\Delta}\right) E\left[\exp\left(\frac{-\epsilon \|X\|_1}{k\Delta}\right)\right] \stepcounter{equation}\tag{\theequation}\label{eq:this_i_want_to_bound}\\\geq& \gamma \exp\left(\frac{-\epsilon \|y\|_1}{k\Delta}\right) \exp(-3(1+\log 2)\epsilon) 
\end{align*}
where the last inequality is by \Cref{lem:exponential_expectations_bound}; we used that $\frac{\epsilon \|X\|_1}{k\Delta}$ has subexponential diameter $\leq 3\epsilon \leq 1/2$, since we have by \Cref{lem:sub_inequality} that $\sigma_{se}[\|X\|] \leq 3 k \Delta$. In the case $k=1$, we simply have that $\sigma_{se}[\frac{\epsilon \|X\|_1}{k\Delta}] \leq \epsilon$, in which case the final bound on \eqref{eq:this_i_want_to_bound} is 
\[
\geq \gamma \exp\left(\frac{-\epsilon \|y\|_1}{k\Delta}\right) \exp(-(1+\log 2)\epsilon) 
\]

At the same time, $f_Y(y) = \gamma \exp(-\frac{\epsilon \|y\|_1}{k\Delta})$ and thus
\begin{align*}
\frac{f_{X+Y}(y)}{f_Y(y)} \geq& e^{-3(1+\log 2)\epsilon}& \text{ if }k>1\\
\frac{f_{X+Y}(y)}{f_Y(y)} \geq& e^{-(1+\log 2)\epsilon}& \text{ if }k=1
\end{align*}
Similarly, we can bound
\begin{align*}
f_{X+Y}(y) =& \gamma E\left[\exp\left(\frac{-\epsilon \|X-y\|_1}{k\Delta}\right)\right] \\\leq& \gamma E\left[\exp\left(\frac{-\epsilon (\|y\|_1-\|X\|_1)}{k \Delta}\right)\right] \\=& \gamma \exp\left(\frac{-\epsilon \|y\|_1}{k\Delta}\right) E\left[\exp\left(\frac{\epsilon \|X\|_1}{k\Delta}\right)\right] \\\leq& \gamma \exp\left(\frac{-\epsilon \|y\|_1}{k\Delta}\right) \exp(9\log(2) \epsilon)
\end{align*}
in the case $k>1$ and
\[
f_{X+Y}(y) \leq \gamma \exp\left(\frac{-\epsilon \|y\|_1}{k\Delta}\right) \exp(3\log(2) \epsilon)
\]
in the case $k=1$. Thus, we have 
\begin{align*}
\frac{f_{X+Y}(y)}{f_Y(y)} \leq& e^{9\log(2) \epsilon} & \text{ if }k>1\\
\frac{f_{X+Y}(y)}{f_Y(y)} \leq& e^{3\log(2) \epsilon} & \text{ if }k=1
\end{align*}
\end{proof}

%

We now prove a lemma that says that in general, if we have a countable number of mechanisms and releasing any fixed $k$ of them is differentially private, then picking the mechanisms adaptively will not violate differential privacy. The proof roughly follows the outline of the proof of adaptive composition \cite{rogers2016privacy}. In what follows, we again use the subscript $\cdot\,_r$ to denote a random bitstring used as the source of randomness of the mechanisms. {\color{red} NOTE TO REVIEWERS: The reviewers may wish to skip this proof and go to page 11 instead as the following proof technique is quite standard. On page 11, we finish the proof of the main theorem.}


\begin{lemma} \label{lem:adaptive}
Let us have a countable number of mechanisms $M_{1,r}, \cdots$%
, such that releasing the value $(M_{i_1,r}(D), \cdots, M_{i_k,r}(D))$ is $\epsilon$-differentially private for any fixed $i_1, \cdots, i_k \in [n]$. 
Then the mechanism $(M_{j_1,r}(D), \cdots, M_{j_k,r}(D))$ is $\epsilon$-differentially private for $j_1, \cdots, j_k \in [n]$ such that $j_\ell$ for $\ell \in [k]$ is drawn from a distribution which is a function of $M_{j_1,r}(D), \cdots, M_{j_{\ell-1},r}(D)$.
\end{lemma}
\begin{proof}
Let $V$ be the adversary selecting the values of $j_\ell$. At the end, $V$ releases $M_{j_1,r}(D), \cdots, M_{j_k,r}(D)$; we prove $V$ is differentially private. At step $\ell$, $V$ picks a distribution as a function of $M_{j_1,r}(D), \cdots, M_{j_{\ell-1},r}(D)$ and then it samples $j_\ell$ from that distribution. We assume that it does this sampling by performing independent unbiased coin flips, and using them to simulate the distribution in question by a standard approach which uses with probability $1$ finite number of flips. We call the sequence of the coin flips $r$. Since the length of $r$ is finite almost surely (one can sample from a countable-support distribution while almost surely using finitely many bits), the distribution of $r$ has countable support up to a set of measure zero; this will help us ensure that we have no issues with measurability. The value of $j_\ell$ then depends deterministically on $r$ and the values $M_{j_{1},r}, \cdots, M_{j_{\ell-1},r}$. If $V$ releases some auxiliary information together with the values $M_{j_1,r}(D), \cdots, M_{j_{\ell-1},r}(D)$, it is only less likely to be private, and we may thus assume that $V$ also releases $r$.

Suppose we fix the randomness $r$ and let $U$ be the range of the mechanisms (if the ranges differ, we set $U$ to be their union). We define an equivalence relation on $U^k$ as follows. We say that $u,v \in U^k$ are equivalent iff they result in the same sequence of $j_1, \cdots, j_k$ (for the fixed value of $r$). Let $s_1(r), s_2(r), \cdots$ be the equivalence classes under of this equivalence relation. The sequence $j_1, \cdots, j_k$ can then be seen as a function of $r$ and of $s_h(r)$ for $h \in \mathbb{N}$. The number of equivalence classes is countable for any $r$. 

Let us have two neighboring databases $D_1 \sim D_2$. For any measurable $Y\subseteq U$ and $t \in \{0,1\}^*$, we can now bound the ratio
{\small
\begin{align*}
\frac{P(r = t \wedge (M_{j_1,r}(D_1), \cdots, M_{j_k,r}(D_1)) \in Y)}{P(r = t \wedge (M_{j'_1,r}(D_2), \cdots, M_{j'_k,r}(D_2)) \in Y)} =& \frac{P(r = t)P((M_{j_1,r}(D_1), \cdots, M_{j_k,r}(D_1)) \in Y | P(r = t))}{P(r = t) P((M_{j'_1,r}(D_2), \cdots, M_{j'_k,r}(D_2)) \in Y| P(r = t))} \\=& \frac{\sum_{h=1}^{\infty} P((M_{j_1,r}(D_1), \cdots, M_{j_k,r}(D_1)) \in Y\cap s_h(t) | P(r = t))}{\sum_{h=1}^{\infty}P((M_{j'_1,r}(D_2), \cdots, M_{j'_k,r}(D_2)) \in Y\cap s_h(t)| P(r = t))} = (*)
\end{align*}
}
The values $j_\ell,j'_\ell$ in $P((M_{j_1,r}(D), \cdots, M_{j_k,r}) \in Y\cap s_h | P(r = t))$ are deterministic as they are a function of $r$ and $s_h$, as noted above. As such, it holds $j_\ell = j'_\ell$. Therefore, by the assumption that $(M_{i_1,r}(D), \cdots, M_{i_k,r}(D))$ is $\epsilon$-differentially private for any fixed $i_1, \cdots, i_k$, it holds that
\[
\frac{ P((M_{j_1,r}(D_1), \cdots, M_{j_k,r}(D_1)) \in Y\cap s_h | P(r = t))}{P((M_{j'_1,r}(D_2), \cdots, M_{j'_k,r}(D_2)) \in Y\cap s_h| P(r = t))} \leq e^{\epsilon}
\]
which allows us to bound $(*)$ as
\[
(*) \leq \frac{\sum_{h=1}^{\infty} e^{\epsilon }P((M_{j'_1,r}(D_2), \cdots, M_{j'_k,r}(D_2)) \in Y\cap s_h| P(r = t))}{\sum_{h=1}^{\infty}P((M_{j'_1,r}(D_2), \cdots, M_{j'_k,r}(D_2)) \in Y\cap s_h| P(r = t))} = e ^ \epsilon
\]
The other inequality holds by symmetry. This proves that $V$ is $\epsilon$-differentially private, as we wanted to prove.
\end{proof}

We are now ready to prove the main theorem of this section. In what follows, we denote by $A_r(\cdot)$ the algorithm $A$ executed with randomness $r$. We formalize the multiple-query setting as having one algorithm which takes as part of its input a query. This differs somewhat from the presentation in the introduction which assumed we have a sequence of algorithms, which we chose there as it required less notation. Note also that while we are not assuming that the algorithm does not know which phase it is (the value of $i$), we may without loss of generality assume this is passed as part of the query.
%
Note that the condition on the queries $x_1, \cdots, x_k$ below simply states that the queries can be chosen adaptively based on the released values, and that they do not have to be deterministic.
Note also that the randomness $r$ must not be released, as the privacy also relies on that randomness.

\begin{theorem} \label{thm:subexponential}
Let us have an algorithm $A(D,x)$ for a database $D$ and a query $x \in U$, where $U$ is a countable universe. Assume there exists a function $g(D,x)$ with its global sensitivity w.r.t.\ $D$ being $\leq \Delta_1$ for any $x$, and such that $\sigma_{se}[A(D,x) - g(D,x)] \leq \Delta_2$ for any $D,x$.

Pick at random $Y_i \sim \text{Laplace}(c(\Delta_1 +  \Delta_2) k /\epsilon))$ for $c = 3+12\log 2$ and for $\epsilon \leq c/6$ and pick $r$ independently uniformly on $\{0,1\}^\infty$. Then for queries $x_1, \cdots, x_k \in U$ where the query $x_i$ is drawn from a distribution that is a function of $A_r(D,x_1) + Y_1, \cdots, A_r(D,x_{i-1}) + Y_{i-1}$, releasing $(A_r(D,x_1) + Y_1, \cdots, A_r(D,x_k) + Y_k)$ is $\epsilon$-differentially private, with the privacy also being over the randomness of $r$.

If $k=1$, then $c = 1+4\log 2$ and $\epsilon \leq c/2$ is sufficient.
\end{theorem}
\begin{proof}
Let us have two neighboring databases $D_1,D_2$. Let $Y = (Y_1, \cdots, Y_k)$, and for $x = (x_1, \cdots, x_k)$, let $A'(D,x) = (A_r(D,x_1) + Y_1, \cdots, A_r(D,x_k) + Y_k)$ for $r$ uniform on $\{0,1\}^\infty$. Similarly, let $g(D, x) = (g(D, x_1), \cdots, g(D,x_k))$. We prove that for any fixed (non-adaptive) queries $x = (x_1, \cdots, x_k$) it holds $f_{A'(D_1,x)}(y)/f_{A'(D_2,x)}(y) \leq e^\epsilon$. If we prove this, the theorem follows: the inequality $f_{A'(D_1,x)}(y)/f_{A'(D_2,x)}(y) \geq e^{-\epsilon}$ holds by symmetry and these bounds together imply $\epsilon$-differential privacy. \Cref{lem:adaptive} then imply that $A'$ is differentially private even for adaptive queries.

Let $R_1 = A'(D_1,x) - g(D_1,x)$ and $R_2 = A'(D_2,x) - g(D_1,x)$ (note the asymmetry in the definitions). We are assuming $R_1$ has subexponential diameter $\leq \Delta_2$. By \Cref{fact:adding_constant}, it holds that $R_2$ has subexponential diameter $\leq \Delta_1/log(2)+\Delta_2$. We now have from \Cref{lem:subexponential} the following bounds
\begin{align*}
\frac{f_{g(D_1,x) + R_1 + Y}(y)}{f_{g(D_1,x) + Y}(y)} = \frac{f_{R_1 + Y}(y-g(D_1,x))}{f_{Y}(y-g(D_1,x))} \leq& e^{9\log(2) \epsilon/c}\\
\frac{f_{g(D_2,x) + R_2 + Y}(y)}{f_{g(D_1,x) + Y}(y)} = \frac{f_{R_2 + Y}(y-g(D_2,x))}{f_{Y}(y-g(D_1,x))} \geq& e^{- 3(1+\log 2) \epsilon/c}
\end{align*}
which in turn allows us to bound
\[
f_{A'(D_1,x)}(y)/f_{A'(D_2,x)}(y) = \frac{f_{g(D_1,x) + R_1 + Y}(x)}{f_{g(D_1,x) + Y}} \cdot \frac{f_{g(D_1) + Y}}{f_{g(D_2) + R_2 + Y}(x)} \leq e^{(3+12\log 2) \epsilon/c} = e^{\epsilon}
\]

If $k=1$, the same computation gives the desired bound for $c = 1+4\log 2$, since \Cref{lem:subexponential} in that case gives gives 
\begin{align*}
\frac{f_{g(D_1,x) + R_1 + Y}(y)}{f_{g(D_1,x) + Y}(y)} \leq& e^{3\log(2) \epsilon/c}\\
\frac{f_{g(D_2,x) + R_2 + Y}(y)}{f_{g(D_1,x) + Y}(y)} \geq& e^{- (1+\log 2) \epsilon/c}
\end{align*}
which again results in the bound $f_{A'(D_1,x)}(y)/f_{A'(D_2,x)}(y) \leq e^\epsilon$.
\end{proof}

\section{Algorithms with bounded mean error}
In this section, we show a weaker version of \Cref{thm:subexponential} that only requires that the error has some number of bounded moments, instead of requiring that it is subexponential. We start by defining the distribution that we will use in our additive noise mechanism.

\begin{definition}
Zero-symmetric Pareto distribution with shape parameter $\alpha>1$ and scale parameter $s>0$, denoted $ZSPareto_\alpha(s)$, is defined by the PDF
\[
\frac{1}{2s} (\alpha-1) (| x|/s +1)^{-\alpha}
\]
\end{definition}
It should be noted that the ``scale parameter" $s$ indeed in some sense represents the scale of the distribution. Specifically, 
\begin{align} \label{eq:pareto_cdf}
\mathbb{P}[|ZSPareto_\alpha(s)| \leq t] = \int_{-t}^t \frac{(\alpha -1) \left(\frac{| x| }{s}+1\right)^{-\alpha }}{2 s} \, dx = 1-\left(\frac{s+t}{s}\right)^{1-\alpha }
\end{align}
which can be made arbitrarily close to $1$ while setting $t = \Theta(s)$, meaning that an arbitrarily large fraction of the probability mass is within $O(s)$ of the origin.

Before we can prove the main theorem of this section, we need to prove a technical lemma. 
\begin{lemma} \label{lem:integral_bound}
Let us have any $0\leq \epsilon\leq 1$, $\alpha > 1$, and $x \geq 0$. It holds
\begin{align} 
\int_0^1 \min\left((1+x)^\alpha,\left|1- \frac{(1-2^{-1/\alpha}) \epsilon}{(1+x)(1-u)^{1/\alpha}}\right|^{-\alpha}\right) du \leq 1+\frac{2\alpha-1}{\alpha-1} \epsilon
\end{align}
\end{lemma}
\begin{proof}
We start by proving that for $0 \leq u \leq 1-\epsilon/(1+x)^\alpha$ and $0 < \epsilon <1$, $\alpha > 1$ and $x, \geq 0$ holds that 
\begin{align} \label{eq:inside_of_integral_inequality}
\left|1- \frac{(1-2^{-1/\alpha})\epsilon}{(1+x)(1-u)^{1/\alpha}}\right|^{-\alpha} \leq 1+\frac{\epsilon}{(1-u)^{1/\alpha}}
\end{align}
We express the inequality using $a$, defined by $u = 1-a \frac{\epsilon}{(1+x)^\alpha}$. The condition $u \leq 1-\frac{\epsilon}{(1+x)^\alpha}$ will then be equivalent to $a \geq 1$. It is thus sufficient to prove that for $a \geq 1$ it holds
\[
\left|1- \frac{(1-2^{-1/\alpha})\epsilon^{1-1/\alpha}}{a^{1/\alpha}}\right|^{-\alpha} \leq 1 + \frac{(1+x) \epsilon^{1-1/\alpha}}{a^{1/\alpha}}
\]
The right-hand side increases with $x$ while the left-hand side is independent of it, and we may thus set $x = 0$ (note that we are assuming $x \geq 0$). It holds that $0 \leq \epsilon^{1-1/\alpha}/a^{1/\alpha} \leq 1$ as $a \geq 1$ and $\epsilon <1$. The inside in the absolute value is thus non-negative, and we may ignore the absolute value. It thus suffices to prove
\[
\left(1- \frac{(1-2^{-1/\alpha})\epsilon^{1-1/\alpha}}{a^{1/\alpha} }\right)^{-\alpha} \leq 1 + \frac{\epsilon^{1-1/\alpha}}{a^{1/\alpha}}
\]
As we said, it holds $0 \leq \epsilon^{1-1/\alpha}/a^{1/\alpha} \leq 1$. By further substituting $b = \epsilon^{1-1/\alpha}/a^{1/\alpha}$, proving our inequality reduces to proving
\begin{align} \label{eq:final_inequality}
(1- (1-2^{-1/\alpha})b)^{-\alpha} \leq 1+b
\end{align}
under the condition $0 \leq b \leq 1$. The left-hand side of \eqref{eq:final_inequality} is convex as a function of $b$ on $[0,1]$: the second derivative is
\[
\frac{\left(2^{1/\alpha }-1\right)^2 \alpha  (\alpha +1) \left(1-\left(1-2^{-1/\alpha }\right) b\right)^{-\alpha }}{\left(-2^{1/\alpha }+2^{1/\alpha } b-b\right)^2}
\]
which is non-negative as the denominator $\left(-2^{1/\alpha }+2^{1/\alpha } b-b\right)^2$ is clearly non-negative, while the numerator is a product of $\left(2^{1/\alpha }-1\right)^2$ which is non-negative, of $\alpha  (\alpha +1)$ which is also non-negative, and of $\left(1-\left(1-2^{-1/\alpha }\right) b\right)^{-\alpha }$ which is also non-negative as $b \leq 1$ and thus $\left(1-2^{-1/\alpha }\right) b < 1$.

The right-hand side of \eqref{eq:final_inequality} is an affine function of $b$. It is thus sufficient to prove \eqref{eq:final_inequality} for $b = 0$ and $b = 1$. The inequality clearly holds for $b = 0$. For $b = 1$, we have $(1-(1-2^{-1/\alpha})b)^{-\alpha} = 2 = 1+ b$, thus proving inequality \eqref{eq:inside_of_integral_inequality}.

We now bound the integral. It holds
\begin{align*}
\int_0^1 \min\left((1+x)^\alpha, \left|1- \frac{(1-2^{-1/\alpha})\epsilon}{(1+x)(1-u)^{1/\alpha}}\right|^{-\alpha}\right) du \leq& \int_0^{1-\epsilon/(1+x)^\alpha} \left|1- \frac{(1-2^{-1/\alpha})\epsilon}{\alpha (1+x)(1-u)^{1/\alpha}}\right|^{-\alpha} du \\&+ \frac{\epsilon}{(1+x)^\alpha}(1+x)^\alpha \\\leq& \int_0^1 1+\frac{\epsilon}{(1-u)^{1/\alpha}} du + \epsilon \\=& 1+ \epsilon + \left[\frac{\alpha u^{\frac{\alpha-1}{\alpha}}}{\alpha-1}\right]_{u=0}^1 \epsilon = 1+\frac{2\alpha-1}{\alpha-1}\epsilon
\end{align*}
where the second inequality uses inequality \eqref{eq:inside_of_integral_inequality}.
\end{proof}

In what follows, we use the notation $\|X\|_p$ for a random variable $X$ to denote the $L_p$ norm $\sqrt[p]{\mathbb{E}[|X|^p]}$.
\begin{lemma}
Let $X$ be a real random variable such that $\|X\|_\alpha \leq \Delta$ and let $Y \sim ZSPareto_\alpha(s)$ for $s = (1-2^{-1/\alpha})^{-1} \Delta /\epsilon$ for $0 \leq \epsilon \leq 1$ and $\alpha > 1$. Consider $y \in \mathbb{R}$. It holds $e^{-(1-2^{-1/\alpha})\alpha\epsilon} \leq f_{X+Y}(y)/f_Y(y) \leq e^{\frac{2\alpha-1}{\alpha-1}\epsilon}$.
\end{lemma}
\begin{proof}
Since $\|X\|_\alpha \leq \Delta$, it holds by the higher-order Chebyshev inequality\footnote{The higher-order Chebyshev inequality states that for $X$ being a real random variable, it holds $\mathbb{P}[|X-\mathbb{E}[X]| \geq t] \leq t^\alpha/\mathbb{E}[|X -  E[X]|^\alpha]$ for any $\alpha \geq 0$.} that $\mathbb{P}[|X| \geq z] \leq \Delta^\alpha/z^\alpha$. Therefore, $F^{-1}_{|X|}(u) \leq \Delta/\sqrt[\alpha]{1-u}$. We will use that $|y-X| \leq |y|+|X|$ and later below also that $|y-X| \geq \max(0,|y|-|X|)$. We start with the simpler case of proving a lower bound.
\begin{align*}
\frac{f_{X+Y}(y)}{f_Y(y)} =& \frac{\mathbb{E}[f_{Y}(y-X)]}{(|y|/s + 1)^{-\alpha}} \stepcounter{equation}\tag{\theequation}\label{eq:second_lemma_1}\\=& \frac{\mathbb{E}[(|y-X|/s + 1)^{-\alpha}]}{(|y|/s + 1)^{-\alpha}} \\\geq& E\left[\left(\frac{|y|/s + |X|/s + 1}{|y|/s + 1}\right)^{-\alpha}\right] \\=& E\left[\left( 1 + \frac{|X|/s}{|y|/s + 1}\right)^{-\alpha}\right] \\\geq& \mathbb{E}[( 1 + |X|/s)^{-\alpha}] \\\geq& \mathbb{E}[1- \alpha|X|/s] 
\\=& 1 - \alpha \mathbb{E}[|X|]/s  \\\geq& 1- \frac{\alpha(1-2^{-1/\alpha}) \Delta }{\Delta} \epsilon \stepcounter{equation}\tag{\theequation}\label{eq:second_lemma_3}\geq e^{-(1-2^{-1/\alpha})\alpha\epsilon}
\end{align*}
where \eqref{eq:second_lemma_1} holds by \Cref{fact:important_identity}, and 
\eqref{eq:second_lemma_3} uses that $\mathbb{E}[|X|] = \|X\|_1 \leq \|X\|_\alpha \leq \Delta$. We now show the upper bound part. We prove an upper bound in terms of the integral which we have bounded in \Cref{lem:integral_bound}.
\begin{align*}
\frac{f_{X+Y}(y)}{f_Y(y)} =&
\frac{\mathbb{E}[(|y-X|/s + 1)^{-\alpha}]}{(|y|/s + 1)^{-\alpha}} \\\leq&
E\left[\min\left((1+|y|/s)^\alpha,\left|\frac{|y|/s - |X|/s + 1}{|y|/s + 1}\right|^{-\alpha}\right)\right] \\=&
E\left[\min\left((1+|y|/s)^\alpha,\left| 1 - \frac{|X|/s}{|y|/s + 1}\right|^{-\alpha}\right)\right] \\=&
E_u\left[\min\left((1+|y|/s)^\alpha,\left| 1 - \frac{F_{|X|}^{\smash{-1}}(u)/s}{|y|/s + 1}\right|^{-\alpha}\right)\right] \\\leq&
E_u\left[\min\left((1+|y|/s)^\alpha,\left| 1 - \frac{(1-2^{-1/\alpha})^{-1}\epsilon}{(|y|/s + 1)(1-u)^{1/\alpha}}\right|^{-\alpha}\right)\right] \\=&
\int_0^1 \min\left((1+|y|/s)^\alpha,\left| 1 - \frac{(1-2^{-1/\alpha})^{-1}\epsilon}{(|y|/s + 1)(1-u)^{1/\alpha}}\right|^{-\alpha}\right) du \\\leq&  1+\frac{2\alpha-1}{\alpha-1} \epsilon \leq e^{\frac{2\alpha-1}{\alpha-1}\epsilon}&
\end{align*}
where we have proven the inequality second to last in \Cref{lem:integral_bound}.
\end{proof}

We are now ready to prove the main theorem of this section.
\begin{theorem} \label{thm:polynomial}
Let us have an algorithm $A(D)$ such that there exists a function $g(D)$ with global sensitivity $\leq \Delta_1$ w.r.t.\ $D$ for which for any input $D$, it holds $\mathbb{E}[|A(D) - g(D)|^\alpha] \leq \Delta_2^\alpha$ for some $\alpha > 1$. Let $Y \sim ZSPareto_\alpha(c (\Delta_1 + \Delta_2)/ \epsilon)$ for $c=\alpha + 2 + 1/(\alpha-1)$ and $\epsilon \leq c$, independent of the randomness of $A$; then $A(D) + Y$ is $\epsilon$-differentially private with respect to $D$.
\end{theorem}
\begin{proof}
This proof follows the strategy of the proof of \Cref{thm:subexponential}. Let us have two neighboring databases $D_1,D_2$. We again prove that for any $y$, it holds $f_{A(D_1)}(y)/f_{A(D_2)}(y) \leq e^\epsilon$; this implies the theorem. We also again set $R_1 = A(D_1) - g(D_1)$ and $R_2 = A(D_2) - g(D_1)$. We are assuming $\|R_1\|_\alpha \leq \Delta_2$ and by the triangle inequality, we have that $\|R_2\|_\alpha \leq \Delta_1+\Delta_2$. Therefore, we have
\begin{align*}
\frac{f_{g(D_1) + R_1 + Y}(y)}{f_{g(D_1) + Y}(y)} = \frac{f_{R_1 + Y}(y-g(D_1))}{f_{Y}(y-g(D_1))} \leq& \exp\left(\frac{2\alpha-1}{(\alpha-1)c}\epsilon\right)\\
\frac{f_{g(D_2) + R_1 + Y}(y)}{f_{g(D_1) + Y}(y)} = \frac{f_{R_1 + Y}(y-g(D_2))}{f_{Y}(y-g(D_1))} \geq& \exp\left(-(1-2^{-1/\alpha})\alpha\epsilon/c  \right)
\end{align*}
which in turn allows us to bound
\begin{align*}
f_{A(D_1)+Y}(y)/f_{A(D_2)+Y}(y) =& \frac{f_{g(D_1) + R_1 + Y}(y)}{f_{g(D_1) + Y}(y)} \cdot \frac{f_{g(D_1) + Y}(y)}{f_{g(D_2) + R_1 + Y}(y)} \\\leq& \exp\left(\left(\frac{(2\alpha-1)}{(\alpha-1)} + (1-2^{-1/\alpha})\alpha\right) \epsilon/c\right) \leq e^{\epsilon}
\end{align*}
where we now argue the last inequality; that will conclude the proof. If we set $c = 2+1/(\alpha-1)+\alpha-2^{-1/\alpha}\alpha$, the last inequality would be an equality. By monotonicity, it is thus sufficient to prove that $2+1/(\alpha-1)+\alpha-2^{-1/\alpha}\alpha \leq 2 + \log 2 + 1/(\alpha-1)$. This is equivalent to $2^{-1/\alpha} \alpha \geq \alpha - \log 2$, which in turn can be re-written as $2^{-1/\alpha} \geq 1-\log(2)/\alpha$. This follows from the inequality $e^x \geq 1+x$ for $x = -\log(2)/\alpha$.
%
%
%
%
%
%
\end{proof}


\section{Implications of our results} \label{sec:applications}
In this section, we give several implications of \Cref{thm:subexponential} and \Cref{thm:polynomial}. This list is by no means meant to be exhaustive. We start with the more straightforward applications and focus on the more involved ones later.

Recall that, as we discussed, the goal in the sublinear setting is not to simply add small amount of noise, but rather to achieve a given level of error as efficiently as possible while guaranteeing differential privacy. This is so because in this setting, the amount of error coming from the algorithm not being exact tends to be much greater than the amount of noise needed to achieve privacy when not subject to having limited resources.

\subsection{The general approach}
In all applications, we take a known algorithm for a given problem, and use either \Cref{thm:subexponential} or \Cref{thm:polynomial} to argue that adding noise to the algorithm's answer ensures privacy. 

Assume the original algorithm had complexity $T(n, \rho)$ and assume for example that the error is of magnitude $\rho n$, namely that for error $R$, it holds $E[R^2]^{1/2} \leq O(\rho n)$ (this can be generalized to $\leq O(\rho f(x))$ for $x$ being the input and $f$ being any function).
We run the algorithm with parameter $\rho' = \epsilon \rho$ and add noise of magnitude $O(\rho n)$ (more generally $O(\rho f(x))$). By \Cref{thm:polynomial} with $\alpha = 2$, as long as the approximated function's sensitivity is $\Delta \leq \epsilon \rho n$, this is $\epsilon$-differentially private\footnote{This upper bound on the sensitivity ensures that $\Delta_2$ in \Cref{thm:polynomial} dominates.}. At the same time, the error is $\leq O(\rho n)$ with arbitrarily high constant probability\footnote{The error coming from the algorithm will not be too large by Chebyshev and the error from the noise will not be too large by our bound \eqref{eq:pareto_cdf}. By the union bound, neither of the two sources of error will be too large with arbitrarily high constant probability.}. The time complexity is $T(n, \epsilon \rho)$.

If we want to achieve a failure probability of $\beta$, we run this algorithm $\Theta(\log \beta^{-1})$ times and take the median. By a standard probability amplification argument, the success probability will be as desired. To achieve $\epsilon$-differential privacy by composition, we have to divide the privacy budget between the runs, resulting in complexity $O(T(n, \epsilon \rho / \log \beta^{-1}) \log \beta^{-1})$. This can be summarized (and generalized with the general function $f(x)$) as follows:
\begin{lemma} \label{lem:privatize_poly}
Suppose there is an algorithm approximating a function $g$ with global sensitivity $\Delta$ such that $E[(A(x) - g(x))^2]^{1/2} \leq \rho f(x)$ for some function $f$ with time/space/query complexity $T(n, \rho)$. Then for $\epsilon \leq O(1)$ there exists an $\epsilon$-differentially private algorithm $A'$ such that when $\epsilon \rho \geq \Omega(\Delta /f(x))$, it holds $P[|A'(x) - g(x)| > \rho f(x)] \leq \beta$ and that has complexity $O(T(n, \epsilon \rho / \log \beta^{-1}) \log \beta^{-1})$.
\end{lemma}



A more efficient approach for decreasing failure probability exists in the case of subexponential error.
Assume the same setting as above, except that the error's subexponential diameter is $\rho n$ instead having only moment bounds (like above, we can generalize to $\rho f(x)$ instead of $\rho n$).
We run the algorithm with parameter $\rho' = \epsilon \rho / \log \beta^{-1}$ and add noise of magnitude $\Theta(\rho n/ \log \beta^{-1})$. This algorithm is $\epsilon$-differentially private by \Cref{thm:subexponential}, as long as $\Delta \leq \epsilon \rho n$. The noise has subexponential diameter $O(\rho' n)$ and by \Cref{lem:sub_inequality}, the total error will up to a constant have the same subexponential diameter. By the definition of subexponential diameter, the probability that the error is $\geq \Theta(\rho n)$ is $\leq \beta$ as desired. This results in complexity $O(T(n, \epsilon\rho/\log \beta^{-1}))$ saving us one $\log \beta^{-1}$ factor. We can summarize this as
\begin{lemma}
Suppose there is an algorithm approximating a function $g$ with global sensitivity $\Delta$ such that $\sigma_{se}[A(x) - g(x)] \leq \rho f(x)$ for some function $f$ with time/space/query complexity $T(n, \rho)$. Then for $\epsilon \leq O(1)$, there exists an $\epsilon$-differentially private algorithm $A'$ such that when $\epsilon \rho \geq \Omega(\Delta/f(x))$, it holds $P[|A'(x) - g(x)| > \rho f(x)] \leq \beta$ and that has complexity $O(T(n, \epsilon \rho / \log \beta^{-1}))$.
\end{lemma}


\noindent
We are now ready to give private algorithms for specific problems.

\subsection{Frequency moment \texorpdfstring{$F_2$}{F\_2}}
In their seminal paper, \citet{alon1996space} show a sketch that allows one to estimate the $F_2$ frequency moment, defined as $F_2(x_1, \cdots, x_n) = \sum_{i=1}^n x_i^2$. In the streaming setting, the vector $x_1, \cdots, x_n$ is given through a stream of updates $y_1, \cdots, y_k$ of the form $y_j = (\ell_j, D_j)$ where $D$ can be negative and we define $x_i = \sum_{j=1}^k I[\ell_j = i] D_j$. Two inputs are then adjacent if they differ in one value $y_j$ for some $j$. The algorithm from \cite{alon1996space} uses space $O(\frac{1}{\rho^2})$ and has mean squared error of $\leq \rho^2 F_2^2 \leq \rho^2 n^4$. The sensitivity of the $F_2$ moment is $n$. This implies the following 
\begin{corollary}
For $\epsilon \leq O(1)$ and $\rho \leq 1/(\epsilon n)$, there is an $\epsilon$-differentially private algorithm that returns an additive $\pm \rho n^2$ approximation of the frequency moment $F_2$ with probability $1-\beta$, and has space complexity $O(\frac{\log^3 \beta^{-1}}{\rho^2 \epsilon^2})$.
\end{corollary}

\noindent No streaming algorithm for private $F_2$ moments was previously known. Shortly after releasing this paper, an approach with the incomparable complexity of $O(\log (n) \log^3 \beta^{-1} / \rho ^2)$ was shown (for $\epsilon$ being not too small) that also has multiplicative approximation guarantees and works (with some loss in the complexity) in the continual release setting \cite{epasto2023differentially}. The setting of $F_p$ for $p \in [0,1]$ has been considered in \cite{wang2021differentially}.


\subsection{Connected components}
\citet{berenbrink2014estimating} have shown an algorithm that returns an estimate $\hat{c}$ of the number of connected components $c$ of a simple graph in time $O(\frac{1}{\rho^2} \log \frac{1}{\rho})$ and has mean squared error $\mathbb{E}[(\hat{c} - c)^2] \leq \rho^2 n^2$. At the same time, the number of connected components has global sensitivity $1$ with respect to edge additions/deletions. This gives us the following
\begin{corollary}
For $\epsilon \leq O(1)$ and $\rho \leq 1/(\epsilon n)$, there is an $\epsilon$-edge-differentially private algorithm that returns an additive $\pm \rho n$ approximation of the number connected components with probability $1-\beta$, and has complexity $O(\frac{\log^3 \beta^{-1}}{\rho^2 \epsilon^2} \log \frac{\log \beta^{-1}}{\rho \epsilon})$.
\end{corollary}
\noindent No private sublinear-complexity algorithm for estimating the number of connected components was previously known.

\subsection{Maximum matching}
\citet{yoshida2009improved} show an algorithm that can approximate the size of the maximum matching to within multiplicative $1+\rho$ in time $d^{O(1/\rho^2)}(1/\rho)^{O(1/\rho)}$ for $d$ being the maximum degree of the input graph. It works by implementing an oracle for a matching of size within factor $1+\rho/2$ of the maximum matching; for a specified vertex, this oracle answers whether the vertex is matched in the oracle's matching. The algorithm then samples $\Theta(1/\rho^2)$ vertices and checks the fraction that is matched in the oracle's matching. The error coming from the oracle is $\leq \rho n/2$ in the worst case and thus has subexponential diameter $O(\rho n)$. The error coming from the sampling has subexponential diameter $O(\rho n)$ by the Hoeffding inequality. By \Cref{lem:sub_inequality}, the subexponential diameter of the error is thus $O(\rho n)$. At the same time, the global sensitivity of the maximum matching size is $\leq 1$ with respect to the removal of one vertex. This gives us the following
\begin{corollary}
For $\epsilon \leq O(1)$ and $\rho \leq 1/(\epsilon n)$, there is an $\epsilon$-node-differentially-private algorithm that returns an additive $\pm \rho n$ approximation of the maximum matching size with probability $1-\beta$ in time $d^{O(\log ^2(\beta^{-1})/(\rho^2 \epsilon^2))}/(\rho\epsilon)^{O(\log (\beta^{-1})/(\rho \epsilon))}$.
\end{corollary}

\noindent This solves the open problem posed in \cite{blocki2022privately} where the authors show a hybrid $(2, \rho n)$ approximation, while we give a purely additive $\pm \rho n$ approximation.

\subsection{Rank queries} \label{sec:rank}
\citet{karnin2016optimal} develop a sketch of size $O(\frac{1}{\rho})$ that allows one to answer rank queries with error with subexponential diameter $\rho n$. We show how to use their sketch to answer range queries over an ordered universe. For small number of queries, this improves upon the work of \cite{kaplan2021note} which has a logarithmic dependency on the universe size. 
This gives us the following corollary.
\begin{corollary} \label{cor:kll}
There is a sketch that allows $\epsilon$-differentially private algorithm that returns $k$ rank queries (potentially adaptive) for $\epsilon \leq O(1)$ with an additive $\pm \rho n$ error with probability $1-\beta$, and has complexity $O(\frac{k \log^2(k/\beta)}{\rho \epsilon})$. 
\end{corollary}
\begin{proof}
We use the KLL sketch with error parameter $\rho' = \rho \epsilon / (k \log (k/\beta))$. This means that the error has a subexponential diameter of $\leq \rho \epsilon n/(k \log (k/\beta))$. Therefore, by \Cref{thm:subexponential}, it holds that using for each query a Laplace mechanism with error magnitude $\Theta(\rho n/\log (k/\beta))$ will result in $\epsilon$-differential privacy.

By \Cref{lem:sub_inequality}, the overall subexponential diameter the error of each answer is $O(\rho n / \log (k/\beta))$ and there the probability of error $O(\rho n)$ is $1-\beta/k$. By the union bound, the overall success probability is $\geq 1-\beta$.
\end{proof}

This is in comparison to the approach of \citet{kaplan2021note} which results in space complexity $O(\frac{\log |U| \log (k/\beta)}{\rho \epsilon})$, improving by a factor of $\log |U|$ for constant $k,\beta$, where $\mathcal{U}$ is the universe.

\subsection{Relative approximation sublinear-time algorithms} \label{sec:sublinear_geom_search}
A common way of designing sublinear time algorithms with relative error is to design an algorithm that requires ``advice" in the form of a constant-factor approximation of the answer and then getting rid of the need for this advice \cite{eden2018approximating}. While this may seem impossible, given that the algorithm needs to roughly know the correct answer in order to give a correct estimate, there is a way that makes this possible under some assumptions. We now show that this advice-removal technique can be modified to also handle differential privacy under mild assumptions.%

First, we need a simple lemma.
\begin{lemma} \label{lem:expectation of median}
Let $X_1, \cdots, X_{2k-1}$ be i.i.d. non-negative random variables. Then 
\[
\mathbb{E}[\text{median}(X_1, \cdots, X_{2k-1})] \leq 2 \binom{2k-1}{k} \mathbb{E}[X_1]
\]
\end{lemma}
\begin{proof}\strut

\vspace{-3.7em}
\begin{align*}
\mathbb{E}[\text{median}(X_1, \cdots, X_{2k-1})] =& \int_0^\infty \mathbb{P}[\text{median}(X_1, \cdots, X_{2k-1}) \geq t] dt \\\leq& \int_0^\infty \binom{2k-1}{k} \mathbb{P}[X_1 \geq t]^{k} dt \\\leq& \binom{2k-1}{k} \int_0^\infty \min(1,\mathbb{E}[X_1]^k/t^k) dt \\=& \binom{2k-1}{k}\left(\mathbb{E}[X_1] + \mathbb{E}[X_1]^k \int_{\mathbb{E}[X_1]}^\infty 1/t^k dt\right) \\=& \binom{2k-1}{k}\left(\mathbb{E}[X_1] + \mathbb{E}[X_1]^k \frac{\mathbb{E}[X_1]^{1-k}}{k-1}\right) \\\leq& 2 \binom{2k-1}{k} \mathbb{E}[X_1]
\end{align*}
where the first inequality is by union bound, over all subsets of size $k$ and the second inequality holds by the Markov inequality.
\end{proof}

We are now ready to prove the theorem. The value $y$ in the following is the advice and intuitively speaking, one should think that we want $y$ to be $\approx g(x)$.
\begin{theorem} \label{lem:geom_search_for_sublinear_alg}
Let us have an algorithm $A(x,y, \rho)$, a function $g(x)$ with global sensitivity $\Delta$ such that $\sup_x g(x) \leq M$ and $\inf_x g(x) \geq m$. Assume that 
\begin{enumerate}[1)]
    \item $\mathbb{E}[A(x,y, \rho)] \leq g(x)$,
    \item $\mathbb{E}[|A(x,y, \rho) - g(x)|] \leq \rho y$,
    \item has complexity $O(T(x,y,\rho))$ such that $O(T(x,y',\rho) \cdot (y'/y)^{c_1}) \leq T(x,y,\rho) \leq O(T(x,y',\rho) \cdot (y'/y)^{c_2})$ for any $x,y', y,\rho$ and some constants $c_1,c_2>0$.
\end{enumerate}
Then for $\epsilon \leq O(1/\log(M/n))$ and $\rho \leq \Delta/(\epsilon n)$, there is an $\epsilon$-differentially private algorithm $A'(x)$ such that $\mathbb{P}[|A'(x) -g(x)| \geq \Theta(\rho g(x))] \leq 1/3$ with complexity $O(T(x,g(x),\epsilon/\log(M/m))+T(x,g(x),\epsilon \rho ))$. 
\end{theorem}
\begin{proof}
Let us denote with $B(x,y,\rho)$ independently executing $A(x,y,\rho)$ five times and taking the median. By the moment amplification lemma (\Cref{lem:moment_amplification}), we have that
\begin{align*}
\mathbb{E}[|B(x,y,\rho)-g(x)|^3] \leq 10 \rho^3 y^3
\end{align*}
Therefore, by \Cref{thm:polynomial}, $B(x,y,\rho) + ZCPareto_3(55 \rho y /\epsilon)$ is $\epsilon$-differentially private.

We now describe the algorithm $A'$. We set $\tilde{y} = M$, and in each iteration, we will decrease $\tilde{y}$ by a factor of $2$. In each iteration, we compute $B(x,\tilde{y},\frac{\epsilon}{330 \cdot 3^{c_2}\log_2(M/m)}) + ZCPareto_3(\frac{\tilde{y}}{3\cdot 3^{c_2}})$. When the returned value is $\geq \tilde{y}$, we compute and return $B(x,\tilde{y}/160,\rho \epsilon) + ZCPareto_3(\rho \tilde{y})$.

The algorithm is $\epsilon$-differentially private by composition, as we executed at most $\log_2(M/m)$ times an algorithm that was $\epsilon/(2\log_2(M/m))$ differentially private, and once an algorithm that was $\epsilon/2$-differentially private.

We now argue correctness. 
It holds $\mathbb{E}[ZCPareto_3(55\rho y/\epsilon)] = 0$ for any $\epsilon$, and by \Cref{lem:expectation of median}, we have $\mathbb{E}[B(x,y,\rho)] \leq 20 g(x)$. Therefore, 
\[
\mathbb{E}[B(x,y,\rho) + ZCPareto_3(55\rho y/\epsilon)] \leq 20 g(x)
\]
for any $\epsilon$. By the Markov inequality, the probability that $B(x,\tilde{y},\frac{\epsilon}{330 \cdot 3^{c_2}\log_2(M/m)}) + ZCPareto_3(\frac{\tilde{y}}{3 \cdot 3^{c_2}}) \geq z$ is then at most $20 g(x)/z$. The probability that we stop with $y \geq 160 g(x)$ is thus at most
\begin{align*}
\sum_{i=1}^{\lceil\log_2(M/m)\rceil} \mathbb{I}[M/2^i \geq 160 &g(x)] \mathbb{P}[B(x,y,\rho) + ZCPareto_3(55 \log_2(M/m)/\epsilon) \geq M/2^i] \\\leq& \sum_{j=0}^\infty \mathbb{P}[B(x,y,\rho) + ZCPareto_3(55 \log_2(M/m)/\epsilon) \geq 160 \cdot 2^j g(x)] \\\leq& \sum_{j=0}^\infty 2^{-j/8} = 1/4
\end{align*}
Therefore, with probability at least $3/4$, we have in the last iteration that $\tilde{y}/160 \leq g(x)$. We call this event $\mathcal{E}$. On $\mathcal{E}$, we have 
\begin{align*}
\mathbb{E}[|B(x,\tilde{y}/160,\epsilon \rho)| + ZCPareto_3(\rho y) - g(x)|] \leq& \mathbb{E}[|B(x,\tilde{y}/160,\epsilon \rho)| - g(x)|] + \mathbb{E}[|ZCPareto_3(\rho y)|] \\\leq& O(\epsilon \rho g(x)) + O(\rho g(x)) = O(\rho g(x))
\end{align*}
where the first term comes from the assumption $2)$ and the second is because $\mathbb{E}[|ZCPareto_3(a)|]= a$ \footnote{This can be easily checked: $\mathbb{E}[|ZCPareto_3(a)|] = \int_0^\infty \mathbb{P}[|ZCPareto_3(a)|\geq t] dt = [-\frac{(a+t) \left(\frac{a+t}{s}\right)^{3} (2 t+a)}{a}]_{t=0}^\infty = a$.}. On $\mathcal{E}$, we thus have by the Markov inequality that the error is $O(\rho g(x))$ with probability $1-1/12$ (or any probability arbitrarily close to $1$). Adding the probability of $\neg \mathcal{E}$, we have that the error is $O(\rho g(x))$ with probability $\geq 2/3$, as we wanted to prove.

It remains to argue time complexity. Suppose $\tilde{y} \leq g(x)$. Then 
\begin{align*}
\mathbb{E}[|B(x,\tilde{y},\frac{\epsilon}{330\cdot 3^{c_2} \log_2(M/m)})& + ZCPareto_3(\tilde{y}/(3\cdot 3^{c_2}))-g(x)|] \\\leq& \mathbb{E}[|B(x,\tilde{y},\frac{\epsilon}{330\cdot 3^{c_2} \log_2(M/m)})-g(x)|] + \mathbb{E}[|ZCPareto_3(\tilde{y}/(3\cdot 3^{c_2}))|] \\\leq& \frac{g(x)}{330\cdot 3^{c_2}} + \frac{g(x)}{3\cdot 3^{c_2}} < \frac{g(x)}{2\cdot 3^{c_2}}
\end{align*}
and by the Markov inequality, it thus holds 
\[
\mathbb{P}[B(x,\tilde{y},\frac{\epsilon}{330\cdot 3^{c_2} \log_2(M/m)}) + ZCPareto_3(\tilde{y}/(3 \cdot 3^c)) \geq g(x)/2] \geq 1-1/3^{c_2}
\]
This means that if $\tilde{y} \leq g(x)/2$, we stop in each iteration with probability $\geq 1- 1/3^{c_2}$. If $\ell$ denotes the number of iterations, then $\mathbb{P}[\ell \geq \lceil\log_2(2M/g(x))\rceil + i] \leq 3^{-c i}$. The expected complexity of the executions of $B(x,\tilde{y},\frac{\epsilon}{330 \cdot 3^{c_2} \log_2(M/m)})$ is thus
\begin{align*}
\sum_{i=0}^{\log(M/m)} T(x,M/2^i,\frac{\epsilon}{330 \cdot 3^{c_2} \log_2(M/m)}) \mathbb{P}[\ell \geq i] \leq& \quad \smashoperator{\sum_{i=0}^{\lceil\log_2(M/g(x))\rceil}}\,\, T(x,M/2^i,\frac{\epsilon}{330 \cdot 3^{c_2} \log_2(M/m)}) \\&+\smashoperator{\sum_{i=\lceil\log_2(2M/g(x))\rceil}^\infty} \,\, (3c)^{-i} T(x,M/2^i,\frac{\epsilon}{330 \cdot 3^{c_2} \log_2(M/m)}) \\\leq& \sum_{j=0}^\infty T(x,g(x)2^{j}/2,\frac{\epsilon}{330 \cdot 3^{c_2} \log_2(M/m)}) \\&+ \sum_{j=0}^\infty T(x,g(x)/2^j,\frac{\epsilon}{330 \cdot 3^{c_2} \log_2(M/m)}) \mathbb{P}[\ell \geq i] \\\leq& \sum_{j=0} T(x,g(x),\frac{\epsilon}{330 \cdot 3^{c_2} \log_2(M/m)}) 2^{-c_1(j-1)} \\&+ \sum_{j=0}^\infty T(x,g(x),\frac{\epsilon}{330 \cdot 3^{c_2} \log_2(M/m)}) 2^{c_2 i} 3^{-c_2 i} \\=& O(T(x,g(x),\frac{\epsilon}{330\cdot 3^{c_2} \log_2(M/m)}))
%
%
\end{align*}
By the exact same argument, the executions of $B(x,\tilde{y}/160,\epsilon \rho)$ contribute $O(T(x,\tilde{y},\epsilon \rho))$, as the computation above is correct even if $T(x,M/2^i,\frac{\epsilon}{330 \cdot 3^{c_2} \log_2(M/m)})$ is replaced by $T(x,M/(160 \cdot 2^i),\rho \epsilon)$. This concludes the proof
\end{proof}

This implies a differentially private algorithm for estimating the average degree of a graph. As the algorithm $A$ of \Cref{lem:geom_search_for_sublinear_alg}, we use the algorithm of \citet{seshadhri2015simpler}. This improves upon the work of \citet{blocki2022privately} who give an algorithm with complexity $\tilde{O}_{\epsilon,\rho}(\sqrt{n})$ under the assumption $m \geq \Omega(n)$.
\begin{corollary}
For $\epsilon \leq O(1)$ and $\rho \leq 1/(\epsilon n)$, there exists an $\epsilon$-edge-differentially private algorithm that returns a $1+\rho$-approximation of the average degree of a graph with probability $1-\beta$ and has complexity $\tilde{O}(\frac{n \log^3 \beta^{-1}}{\epsilon^2 \rho^2 \sqrt{m}})$.
\end{corollary}
\begin{proof}
As the authors prove, their estimator is unbiased.
The variance of their algorithm is analyzed in \cite[Section~3.2]{assadi}, proving that the variance is $O(\rho^2 m^{3/2} \sqrt{y})$ and complexity $O(n/(\rho^2 \sqrt{y}))$ for $y$ being an ``advice parameter'' like in \Cref{lem:geom_search_for_sublinear_alg}. Using \Cref{lem:geom_search_for_sublinear_alg} gives us a private with constant success probability and complexity $\tilde{O}(\frac{n}{\epsilon^2 \rho^2 \sqrt{m}})$. By standard probability amplification (dividing the privacy budget between the executions), we get the desired claim.
\end{proof}

\section{Open problems and conjectures} \label{sec:open_problems}
\paragraph{Improve constants for the noise magnitude.} It seems likely that our analysis is not tight in terms of the constants. Is it possible to improve them? What are the best possible constants? Can the constants be improved for subgaussian error? Is it possible to get improved constants for some specific problems by using a definition of a subexponential diameter which differs by a constant factor\footnote{As we mentioned in \Cref{sec:preliminaries}, there are several definitions of $\sigma_{se}$, that differ by constant factors. See \cite{vershynin2018high} for the definitions.}? Is it possible to improve the constants in the subexponential case, if we assume that the estimator is unbiased (that is, that $\mathbb{E}[A(D,x)] = g(D,x)$)? Can the constants be improved in the case of polynomial tails of the error? Are there any other distributions than the zero-symmetric power-law mechanism, which would achieve differential privacy under the same assumptions, but would have smaller variance/mean absolute deviation?

\paragraph{Lower bounds for multiple queries with bounded error moments.} Suppose we want to publish the values of $k$ dependent random variables, like in the subexponential case, but the error only has a finite number of finite moments. Does there exist a distribution of noise such that adding this noise to the output of the algorithm results in differential privacy? We conjecture that this is not possible and that it is inherent that the zero-symmetric Pareto mechanism only works for a number of variables that depends on the number of finite moments of the error. If this is the case, what is the number of moments, that we need (as a function of $k$) for a multivariate variant of \Cref{thm:polynomial} to hold? If this is some slow-growing function of $k$, it may be feasible to use the median trick to bound higher moments like in \cite{larsen2021countsketches}
, thus allowing a weaker multivariate version of \Cref{thm:polynomial} when only assuming an upper bound on the mean absolute deviation, but with superlinear dependency on $k$.



\paragraph{Normal noise in the case of subgaussian error.} Suppose the algorithm's error distribution is subgaussian (and not just subexponential). Is it possible to use, in that case, the Gaussian mechanism, and get the error scale with $\sqrt{k}$ instead of with $k$ for $k$ being the number of queries?

\paragraph{What if the amount of error depends on the input?} There are some cases in which our approach does not apply. One possible reason is that the amount of error is large in the worst case, but it may be small in a certain instance. Perhaps one could use something akin to smoothed sensitivity? One notable case where this problem arises is the relative estimation of the frequency $F_p$ moments in the streaming setting. In that problem, the amount of error depends on the actual $F_p$ norm, and as such an instance-independent upper bound on the amount of error may not capture well the actual amount of error.



\section*{Acknowledgements}
I would like to thank Rasmus Pagh and anonymous reviewers for helping to improve this paper.

\bibliographystyle{plainnat}
\bibliography{literature}

\appendix
\section{Deferred proofs} \label{sec:appendix}
Proof of \Cref{lem:sub_inequality}
\begin{proof}
By the Markov inequality, we have for any $\lambda > 0$ that
\begin{align} \label{eq:some_equation}
\mathbb{P}[|\sum_{i=1}^k X_i| \geq t] = \mathbb{P}[\exp(\lambda |\sum_{i=1}^k X_i|) \geq e^{\lambda t}] \leq \mathbb{E}[\exp(\lambda |\sum_{i=1}^k X_i|)] e^{-\lambda t}
\end{align}
If we set $\lambda = (3 \sum_{i=1}^k \sigma_{se}[X_i])^{-1}$, then we get
\begin{align*}
E\left[\exp\left(\lambda |\sum_{i=1}^k X_i|\right)\right] =& E\left[\exp\left(\frac{|\sum_{i=1}^k X_i|}{3 \sum_{i=1}^k \sigma_{se}[X_i]}\right)\right] \\\leq& E\left[\exp\left(\frac{\sum_{i=1}^k |X_i|}{3 \sum_{i=1}^k \sigma_{se}[X_i]}\right)\right] \\=& E\left[\exp\left(\sum_{i=1}^k\frac{\sigma_{se}[X_i]}{\sum_{i=1}^k \sigma_{se}[X_i]}\frac{|X_i|}{3 \sigma_{se}[X_i]}\right)\right]
\\\leq& \sum_{i=1}^k \frac{\sigma_{se}[X_i]}{\sum_{i=1}^k \sigma_{se}[X_i]}E\left[\exp\left(\frac{|X_i|}{3 \sigma_{se}[X_i]}\right)\right] \leq (*)
\end{align*}
where the last line is by the Jensen inequality. We have from \Cref{lem:exponential_expectations_bound} that $\mathbb{E}[\exp(\frac{|X_i|}{3 \sigma_{se}[X_i]})] \leq \frac{2^{1/3}}{1-1/3} < 2$, where we have used that $\sigma_{se}[\frac{|X_i|}{3 \sigma_{se}[X_i]}] = 1/3$, giving us
\[
(*) \leq \sum_{i=1}^k \frac{\sigma_{se}[X_i]}{\sum_{i=1}^k \sigma_{se}[X_i]} \cdot 2 = 2\,.
\]
This allows us to use \eqref{eq:some_equation} to bound
\[
\mathbb{P}[|\sum_{i=1}^k X_i| \geq t] \leq 2e^{-t/(3 \sum_{i=1}^k \sigma_{se}[X_i])}
\]
which is equivalent to $\sigma_{se}[\sum_{i=1}^t X_i] \leq 3 \sum_{se}[X_i]$ by the definition of $\sigma_{se}$.
\end{proof}
\end{document}